\RequirePackage{fix-cm}
\documentclass[smallextended]{svjour3}       
\smartqed  
\usepackage[american]{babel}
\usepackage{amsmath}
\usepackage{amssymb}
\usepackage{xspace}
\usepackage{hhline}
\usepackage{multirow}
\usepackage{graphicx}
\usepackage[linesnumbered, ruled,noend]{algorithm2e}
\usepackage{todonotes}
\usepackage{enumerate}
\usepackage{tikz}
\tikzstyle{vertex}=[circle, draw, inner sep=0pt, minimum size=6pt]
\newcommand{\vertex}{\node[vertex]}

\newcommand{\ints}{\ensuremath{\mathbb{Z}}\xspace}

\newcommand{\A}{\ensuremath{\operatorname{\textsc{Alg}}}\xspace}
\newcommand{\aparent}{\ensuremath{\operatorname{\textsc{Parent}}}\xspace}
\newcommand{\alayersparent}{\ensuremath{\operatorname{\textsc{LowParent}}}\xspace}
\newcommand{\aparentplus}{\ensuremath{\operatorname{\textsc{LowParent}}}\xspace}
\newcommand{\ALG}{\ensuremath{\operatorname{\textsc{Alg}}}\xspace}

\newcommand{\greedy}{\ensuremath{\textsc{Greedy}}\xspace}

\newcommand{\opt}{\ensuremath{\operatorname{\textsc{Opt}}}\xspace}
\newcommand{\inc}{\ensuremath{\text{inc}}\xspace}
\newcommand{\onopt}{\ensuremath{\opt^\textsc{\inc}}\xspace}
\newcommand{\onoptT}{\ensuremath{\onopt_T}\xspace}
\newcommand{\onoptC}{\ensuremath{\onopt_C}\xspace}
\newcommand{\offopt}{\ensuremath{\opt^\textsc{off}}\xspace}

\newcommand{\cora}{\ensuremath{\mathbb{CR}}\xspace}
\newcommand{\cron}{\ensuremath{\cora^{\textsc{\inc}}}\xspace}
\newcommand{\croff}{\ensuremath{\cora^{\textsc{off}}}\xspace}

\newcommand{\gammaC}{\ensuremath{\gamma_C}\xspace}
\newcommand{\gammaT}{\ensuremath{\gamma_T}\xspace}
\newcommand{\gammaI}{\ensuremath{\gamma_I}\xspace}

\newcommand{\ds}{\textrm{DS}\xspace}
\newcommand{\cds}{\textrm{CDS}\xspace}
\newcommand{\tds}{\textrm{TDS}\xspace}
\newcommand{\ids}{\textrm{IDS}\xspace}

\newcommand{\SET}[1]{\ensuremath{\left\{#1\right\}}\xspace}
\newcommand{\SETOF}[2]{\left\{#1\mid#2\right\}\xspace}
\newcommand{\EDGE}[2]{\SET{#1,#2}\xspace}
\newcommand{\CEIL}[1]{\lceil #1\rceil}

\newcommand{\cc}{\mathrm{c}}
\newcommand{\incr}{incremental\xspace}

\journalname{Algorithmica}
\begin{document}

\title{Online Dominating Set\,\thanks{A preliminary version of this paper
appeared in the \emph{15th Scandinavian Symposium and Workshops on Algorithm Theory (SWAT)}, LIPIcs, vol.\ 53, Schloss Dagstuhl - Leibniz-Zentrum f\"{u}r Informatik, 2016, pp.~21:1--21:15. Most of the work was done while the fourth author was at the University of Southern Denmark.
The first, third, fourth, and fifth authors were supported in part by the
Danish Council for Independent Research, Natural Sciences,
grants DFF-1323-00247 and DFF-7014-00041,
and the Villum Foundation, grant VKR023219.}}


\author{Joan Boyar \and Stephan J. Eidenbenz \and
Lene M. Favrholdt \and Michal Kotrb\v{c}\'{\i}k \and Kim S. Larsen}

\authorrunning{Boyar, Eidenbenz, Favrholdt, Kotrb\v{c}\'{\i}k,  Larsen} 

\institute{Joan Boyar \at
              Department of Mathematics and Computer Science,
              University of Southern Denmark, Odense, Denmark \\
              \email{joan@imada.sdu.dk}
           \and
           Stephan J. Eidenbenz \at
              Los Alamos National Laboratory, Los Alamos, NM, USA \\
              \email{eidenben@lanl.gov}
           \and
           Lene M. Favrholdt \at
              Department of Mathematics and Computer Science,
              University of Southern Denmark, Odense, Denmark \\
              \email{lenem@imada.sdu.dk}
           \and
           Michal Kotrb\v{c}\'{\i}k \at
              School of Mathematics and Physics,
              University of Queensland, Brisbane, Australia \\
              \email{m.kotrbcik@gmail.com}
           \and
           Kim S. Larsen \at
              Department of Mathematics and Computer Science,
              University of Southern Denmark, Odense, Denmark \\
              \email{kslarsen@imada.sdu.dk}
}

\date{Received: date / Accepted: date}

\maketitle

\begin{abstract}
This paper is devoted to the online dominating set problem and its
variants.
We believe the paper represents the first systematic study of
the effect of two limitations of online algorithms:
making irrevocable decisions while not knowing the future,
and being incremental,
i.e., having to maintain solutions to all prefixes of the input.
This is quantified through competitive analyses of online algorithms
against two optimal algorithms, both knowing the entire input,
but only one having to be incremental.
We also consider the competitive ratio of the weaker of the two optimal
algorithms against the other.

We consider important graph classes,
distinguishing between connected and not necessarily connected graphs.
For the classic graph classes of trees, bipartite, planar, and general
graphs, we obtain tight results in almost all cases.
We also derive upper and lower bounds for the class of bounded-degree graphs.
From these analyses, we get detailed information regarding the 
significance of the necessary requirement that online algorithms be
incremental. In some cases, having to be incremental
fully accounts for the online algorithm's disadvantage.
\end{abstract}

\section{Introduction}
We consider online versions of a number of NP-complete graph problems,
\emph{Dominating Set} (\ds), and variants hereof.
Given an undirected graph $G=(V,E)$ with vertex set~$V$ and edge set~$E$,
a set $D\subseteq V$ is a \emph{dominating set} for~$G$ if for all
vertices $u\in V$, either $u\in D$ (containment) or there exists an edge
$\EDGE{u}{v}\in E$, where $v\in D$ (dominance).
The objective is to find a dominating set of minimum cardinality.

In the variant \emph{Connected} Dominating Set (\cds), we add the requirement
that $D$ be connected (if $G$ is not connected, $D$ should be connected
for each connected component of~$G$).
In the variant \emph{Total} Dominating Set (\tds),
every vertex must be dominated
by another, corresponding to the definition above with the
``containment'' option removed.
We also consider \emph{Independent} Dominating Set (\ids), where we add
the requirement that $D$ be independent, i.e., if $\EDGE{u}{v}\in E$,
then $\SET{u,v}\not\subseteq D$.
In both this introduction and the preliminaries section, 
when we refer to Dominating Set, the statements 
are relevant to all the variants unless explicitly specified otherwise.

The study of Dominating Set and its variants dates back at least to
seminal books by K\"{o}nig~\cite{K50}, Berge~\cite{B62}, and Ore~\cite{O62}.
The concept of domination readily lends itself to
modeling many conceivable practical problems.
Indeed, at the onset of the field,
Berge~\cite{B62} mentions a possible application of keeping all points in a
network under surveillance by a set of radar stations, and Liu~\cite{L68}
notes that the vertices in a dominating set can be thought
of as transmitting stations that can transmit messages to all stations
in the network. Several monographs are devoted to domination~\cite{HHS98},
total domination~\cite{HY13}, and connected domination~\cite{DW13}, and we
refer the reader to these for further details.

We consider \emph{online}~\cite{BE98b} versions of these problems.
More specifically, we consider the vertex-arrival model where the vertices
of the
graph arrive one at a time and with each vertex, the edges connecting
it to previous vertices are also given. 
If the online algorithm decides to include a vertex in the set~$D$,
this decision is irrevocable.
Note, however, that not just a new vertex but also vertices given
previously may be added to~$D$ at any time.
An online algorithm must make this decision without
any knowledge about possible future vertices.
Note that, since an online algorithm does not know the size of the
input graph, it has to maintain a feasible solution at any time.
Since the graph consisting of 
a single vertex does not have a total dominating set at all and
isolated vertices do not dominate any vertices, we allow an online 
algorithm for \tds to \emph{not} include isolated vertices in the solution, 
unlike the other variants of \ds.

Defining the nature of the irrevocable decisions is a
modeling issue, and one could alternatively have made the decision
that also the act of \emph{not} including the new vertex in~$D$ should
be irrevocable, i.e., not allowing algorithms to include already given
vertices in~$D$ at a later time.  The main reason for our choice of
model is that it is much better suited for applications such as
routing in wireless networks for which domination is intensively
studied; see for instance~\cite{DB97} and the citations
thereof. Indeed, when domination models a (costly) establishment of
some service, there is no reason why \emph{not} establishing a service
at a given time should have any inherent costs or consequences, such
as preventing one from doing so later.  Furthermore, the stricter
variant of irrevocability results in a problem for which it becomes
next to impossible for an online algorithm to obtain a non-trivial
result in comparison with an optimal offline algorithm.
Consider, for example, an instance
where the adversary starts by giving a vertex followed by a number
of neighbors of that vertex.  If the algorithm ever rejects one of
these neighbors, the remaining part of the sequence will consist of
neighbors of the rejected vertex and the neighbors must all be selected.  
This shows that, using this model
of irrevocability, online algorithms for \ds or \tds would have to
select at least $n-1$ vertices, while the optimal offline algorithm
selects at most two.
For \cds it is even worse, since rejecting any vertex
could result in a disconnected dominating set.
A similar observation is made in~\cite{KT97} for this model, though they
 focus more on a different model, where the vertices are known in
advance, and all edges incident to a particular vertex are presented
when that vertex arrives.

An online algorithm can be seen as having two characteristics:
it maintains a feasible solution at any time, and it has no knowledge
about future requests.
The first is a consequence of the algorithm not knowing the length of
the sequence.
We also define a larger class of algorithms:
An {\em \incr} algorithm is an algorithm that maintains a feasible
solution at any time.
It may or may not know the whole input from the beginning.

We analyze the quality of online algorithms for the dominating set problems
using \emph{competitive analysis}~\cite{ST85j,KMRS88j}.
Thus, we consider the size of the dominating set an online algorithm
computes up against the result obtained by an optimal offline algorithm, \opt.

As something a little unusual in competitive analysis, we are working
with two different optimal algorithms. This is with the aim of
investigating whether it is predominantly the requirement to maintain
feasible solutions or the lack of knowledge of the future which makes
the problem hard.  Thus, we define \onopt to be an \emph{optimal \incr}
algorithm and \offopt to be an \emph{optimal offline}
algorithm, i.e., it is given the entire input, and then produces a
dominating set for the whole graph. The reason for this distinction is
that in order to properly measure the impact of the knowledge of the
future, it is necessary that it is the sole difference between the
algorithm and \opt. Therefore, \opt has to solve the same problem and
hence the restriction on \onopt.  While such an attention to comparing
algorithms to an appropriate \opt already exists in the literature, to
the best of our knowledge the focus also on the comparison of
different optimum algorithms is a novel aspect of our work. Previous
results requiring the optimal offline algorithm to solve the same
problem as the online algorithm include \cite{BL99} which
considers \emph{fair} algorithms that have to accept a request
whenever possible, and thus require \opt to be fair as well,
\cite{ChSW11} which studies \emph{$k$-bounded-space} algorithms for
bin packing that have at most $k$ open bins at any time and requires
\opt to also adhere to this restriction, and \cite{BSV15} which
analyzes the performance of online algorithms for a variant of bin
packing against a \emph{restricted offline optimum} algorithm that
knows the future, but has to process the requests in the same order as
the algorithm under consideration.

Given an input sequence $I$ and an algorithm \ALG,
we let $\ALG(I)$ denote the size of the
dominating set computed by \ALG on $I$.
Then \ALG is
\emph{$c$-competitive} if there exists a constant $\alpha$ such that
for all input sequences $I$, $\ALG(I)\leq c\opt(I)+\alpha$,
where \opt may be \onopt or \offopt, depending on the context.
The (asymptotic) \emph{competitive ratio} of \ALG is the infimum
over all such~$c$ and we denote this $\cron(\ALG)$
and $\croff(\ALG)$, respectively.
If the inequality above holds without the additive constant $\alpha$, the
algorithm is said to be {\em strictly $c$-competitive}, and the 
 {\em strict competitive ratio} is the infimum over all such~$c$.
When considering competitive ratios that are linear in the input size,
$n$, we will use the strict competitive ratio.
This is mainly to avoid technicalities arising from the fact that if
an algorithm is $n/a$-competitive for some constant $a$, then it is
also $(n/a-b)$-competitive for any constant $b$.

We consider the four dominating set problem variants
on various graph types, including trees, bipartite, and general graphs
and to some extent planar graphs, obtaining tight results
in almost all cases.
We also consider graphs of bounded degree, giving upper and lower
bounds as a function of the maximum degree, $\Delta$.
In all cases, we also consider the online variant where the
adversary is restricted to giving the vertices in such a manner
that the graph given at any point in time is connected.
In this case, the graph is called {\em always-connected}.
One motivation is that graphs in applications such as routing in 
networks are most often connected.

The results for online algorithms are summarized in 
Tables~\ref{table:any-onopt} and~\ref{table:any-offopt}. The strict upper
bound on the competitive ratio against \onopt for general graphs is
$\frac{n+3}{4}$. Note that for this, and other strict competitive
ratios containing~$n$,
we ignore the additive constant (in the table), writing $n/4$ in this case.
The results for \onopt against \offopt are identical to the results of
Table~\ref{table:any-offopt}, except that for \ds on trees,
$\croff(\onopt)=2$, for \ds on always-connected planar graphs,
$\croff(\onopt) = n/2$, and for always-connected bounded-degree graphs, the lower bound that we prove is $\croff(\onopt) \geq (\Delta-1)/2$.
The results are discussed in the conclusion.

\begin{table}
\begin{tabular}{|l||c|c|c|c|}
\hline
Graph class & DS & CDS & TDS & IDS \\
\hline\hline
Trees & 2 & \multicolumn{2}{c|}{$1$} & \multirow{6}{*}{1}\\ 
\hhline{|----|~|}
Bipartite & \multicolumn{2}{c|}{\multirow{2}{*}{$n/4$}} & $n/4$ & \\
\hhline{|-|~~-|~|}
Always-connected bipartite & \multicolumn{2}{c|}{} & $[n/6;2n/9]$ & \\
\hhline{|----|~|}
Bounded-degree & \multirow{2}{*}{$[\Delta/2-1/4; \Delta]$} & $[\Delta/2; \Delta+1]$ &
$[\Delta/2; \Delta]$ & \\ 
\hhline{|-|~|-|-|~|}
Always-conn.\ bounded-degree & & $[\Delta/3; \Delta -1]$ & $[\Delta/3; \Delta]$ & \\
\hhline{|-|-|-|-|~|}
General graphs & \multicolumn{3}{c|}{$n/4$}  & \\
\hline
\end{tabular}
\caption{Bounds on the competitive ratio of any online algorithm
with respect to \onopt.\label{table:any-onopt}}
\end{table}

\begin{table}
\begin{tabular}{|l||c|c|c|c|}
\hline
Graph class & DS & CDS & TDS & IDS \\ 
\hline\hline
Trees & $[2;3]$ & $1$ & 2 & \multirow{3}{*}{$n$}\\ 
\hhline{|----|~|}
Bipartite & \multicolumn{2}{c|}{$n$} & \multirow{2}{*}{$n/2$} & \\
\hhline{|---~|~|}
Always-connected bipartite &    \multicolumn{2}{c|}{$n/2$} & &  \\
\hline
Bounded-degree & $\Delta$ & $\Delta +1$ & \multirow{2}{*}{$[\Delta -1; \Delta]$}  &$\Delta$ \\
\hhline{|---|~|-|}
Always-conn.\ bounded-degree &$[\Delta-2;\Delta]$ & $[\Delta-2;\Delta -1]$ &  &$[\Delta -1; \Delta]$ \\
\hline
Planar & \multicolumn{2}{c|}{\multirow{2}{*}{$n$}} &\multirow{2}{*}{$n/2$} & \multirow{2}{*}{$n$}\\
\hhline{|-|~~|~|~|}
Always-connected planar & \multicolumn{2}{c|}{}  & & \\
\hline
\end{tabular}
\caption{Bounds on the competitive ratio of  any online algorithm  with respect to \offopt.\label{table:any-offopt}}
\end{table}

\section{Preliminaries}
Since we are studying online problems, the order in which vertices
are given is important. We assume throughout the paper that
the indices of the vertices of $G$, $v_1,\ldots, v_n$,
indicate the order in which they are given to the online algorithm,
and we use $\ALG(G)$ to denote the size of the dominating set
computed by \ALG using this ordering.
When no confusion can occur, we implicitly assume that the dominating
set being constructed by an online algorithm $\ALG$ is denoted by $D$. 
We use the phrase \emph{select a vertex} to mean that the vertex
in question is added to the dominating set in question.
We use $G_i$ to denote the subgraph of $G$ induced by $\SET{v_1, \ldots, v_i}$.
We let $D_i$ denote
the dominating set constructed by $\ALG$ after processing the first
$i$ vertices of the input.
When no confusion can occur, we 
sometimes implicitly 
identify a dominating set $D$ and the subgraph it induces. 
For example, we may say that
\emph{$D$ has $k$ components} or \emph{is connected},
meaning that the subgraph of~$G$
induced by~$D$ has $k$ components or is connected, respectively.

Online algorithms must compute a solution for all prefixes of
the input seen by the algorithm, since the input could terminate at
any point. Given the irrevocable decisions,
this can of course affect the possible final sizes of a dominating
set. When we want to emphasize that a bound is derived under this
restriction, we use the word \emph{\incr} to indicate this,
i.e., if we discuss the size of an \incr dominating set $D$ of $G$,
this means that 
$D_1\subseteq D_2\subseteq \cdots \subseteq D_n=D$ and that 
$D_i$ is a dominating set of $G_i$ for each $i$.
Note in particular that any \incr algorithm, including \onopt, for
\ds, \cds, or \ids must select the first vertex.

Throughout the text, we use standard graph-theoretic notation.
In particular, the \emph{path on~$n$ vertices} is denoted~$P_n$.
A {\em star} with $n$~vertices is the complete bipartite graph~$K_{1,n-1}$.
In a rooted tree, an {\em internal} vertex is a vertex that has at
least one child vertex.
For a vertex $v$, $N(v)$ denotes the set of neighbors of $v$.
We  use $\cc(G)$ to denote the number of components of a graph $G$.
The size of a minimum
dominating set of a graph $G$ is denoted by~$\gamma(G)$.
We use indices to
indicate variants, using $\gammaC(G)$, $\gammaT(G)$, and $\gammaI(G)$
for Connected, Total, and Independent Dominating Set, respectively.
This is an alternative notation for the size computed by \offopt.
We also use these indices on $\onopt$ to indicate which
variant is under consideration.
Sometimes, when the problem considered is clear from the context or
we consider more problems at the same time, we may omit the index.
We use $\Delta$ to denote the maximum degree of the graph under
consideration.
Similarly, we always let $n$ denote the number of vertices in the
graph.

In many of the proofs of lower bounds on the competitive ratio, 
when the path, $P_n$, is considered, either as the entire input
or as a subgraph of the input,
we assume that it is given in the \emph{standard order}, the order
where the first vertex given is one of the two endpoints, and each
subsequent vertex 
is a neighbor of the vertex given in the previous step.
When the path is a subgraph of the input graph, we often
extend this standard order of the path to
an \emph{adversarial order} of the input graph -- a fixed ordering of the 
vertices that yields an input attaining the bound.
Typically, the adversarial order consists of a path in the 
standard order, followed by one or more high-degree vertices
off the path.

In some online settings,
we are interested in connected graphs, where the vertices are given
in an order such that the subgraph induced at any point
in time is connected.
In this case, we use the term
\emph{always-connected},
indicating that we are considering a connected graph~$G$, and
all the partial graphs~$G_i$ are connected.
We implicitly assume that trees are always-connected and we drop the adjective.
Since all the classes we consider are hereditary (that is, any induced subgraph
also belongs to the class), no further restriction of partial inputs
$G_i$ is necessary. In particular, these conventions 
imply that for trees, the vertex arriving at any step (except the first) is
connected to exactly one of the vertices given previously, and since we
consider unrooted trees, we can think of that vertex as the \emph{parent} of the
new vertex.

\section{The Cost of Being Online}

In this section, we analyze the the performance of online algorithms
for the four variants of Dominating Set.
We compare the algorithms to \onopt, thus comparing algorithms restricted
to making the same irrevocable decisions, and thereby
investigating the role played by the
(absence of) knowledge of the future.
We also compare the online algorithms to \offopt.

We start with Independent Dominating Set.
\begin{proposition}
For any graph~$G$, there is a unique \incr independent dominating set.
\end{proposition}
\begin{proof}
We fix $G$ and proceed inductively.
The first vertex has to be selected due to the online requirement.
When the next vertex, $v_{i+1}$, is given, if it is dominated by a vertex 
in~$D_i$,
it cannot be selected, since then $D_{i+1}$ would not be independent.
If $v_{i+1}$ is not dominated by a vertex in~$D_i$, then
$v_{i+1}$ or one of its neighbors must be selected.
However, none of $v_{i+1}$'s neighbors can be selected, since if they
were not selected already, then they are dominated, and selecting one
of them would violate the independence criteria.
Thus, $v_{i+1}$ must be selected.
In either case, $D_{i+1}$ is uniquely defined. 
\qed\end{proof}

Since a correct incremental algorithm is uniquely defined by this proposition
by a forced move in every step, \onopt must behave exactly the same.
This fills the column for Independent Dominating Set
in~Table~\ref{table:any-onopt}.

For Dominating Set, Connected Dominating Set, and Total Dominating
Set, we start by using the size of a given dominating set to bound the
sizes of some connected or incremental equivalents. 
The following theorem does not address TDS directly, but in many
cases, it can be applied to this problem as well, since any connected
dominating set including more than one vertex in each connected
component is a total dominating set.

\begin{theorem}
\label{thm:cds}
Let $G$ be always-connected,
let $S$ be a dominating set of $G$,
and let $R$ be an \incr dominating set of $G$.
Then the following hold:
\begin{enumerate}[(i)]
\item \label{thm:cds:c} There is a connected dominating set $S'$ of $G$
such that  $|S'| \le |S| + 2(\cc(S)-1)$.
\item \label{thm:cds:ic} There is an 
\incr connected dominating set $R'$
of $G$ such that $|R'| \le |R| + \cc(R) -1$. 
\item \label{thm:cds:i} If $G$ is a tree, there is an \incr
dominating set $R''$ of $G$ such that $|R''| \le |S| + \cc(S)$.
\end{enumerate}
Moreover, all three bounds are tight for infinitely many graphs.
\end{theorem}
\begin{proof}
To obtain the upper bound of~(\ref{thm:cds:c}), we argue that by selecting
additionally at most $2(\cc(S)-1)$ vertices, we can connect all the
components in $S$. We do this inductively.  If there are two
components that can be connected by a path of at most two unselected vertices,
we select all the vertices on this path and continue
inductively. Otherwise, assume to the contrary that all pairs of
components require the selection of at least three vertices to become
connected.  We choose a shortest such path of length~$k$ consisting
of vertices $u_1,\ldots, u_k$, where $u_i$ is dominated by a
component $C_i$ for all $i$.  If $C_1 \not= C_2$, we can connect
them by selecting $u_1$ and $u_2$, which would be a contradiction.
If $C_1=C_2$, then we have found a shorter path between $C_1$ and
$C_k$; also a contradiction.  We conclude that $|S'| \le |S| + 2(\cc(S) -1)$,
which proves~(\ref{thm:cds:c}).

To see that the bound is tight, consider a path $P_n$ in the standard order,
where $n\equiv 0 \pmod 3$. Clearly, the size of a minimum dominating
set $S$ of $P_n$ 
is $n/3$ and $\cc(S) = n/3$. On the other hand,
the size of any minimum connected dominating set of $P_n$ is $n-2$
and $n-2 = |S| + 2(\cc(S) - 1)$.

To prove~(\ref{thm:cds:ic}), we label the components of~$R$
in the order in which their first vertices
arrive. Thus, let $C_1, \ldots, C_k$ be the components of~$R$,
and, for $1 \leq i \leq k$, let
$v_{j_i}$ be the first vertex of~$C_i$ that arrives. Note that we assume that
$v_{j_i}$ arrives before $v_{j_{i+1}}$ for each $i = 1,\ldots, k-1$.
We prove that for each component~$C_i$ of~$R$,
there is a path of length~$2$
joining $v_{j_i}$ with $C_h$~in $G_{j_i}$ for some $h < i$,
i.e., a path with only one vertex not belonging to either component.
Let $P = v_{l_1}, \ldots, v_{l_m}, v_{j_i}$ be a
shortest path in~$G_{j_i}$ connecting $v_{j_i}$ and some component
$C_h$, $h < i$, and assume for the sake of contradiction that $m \geq 3$.
In $G_{j_i}$, the vertex~$v_{l_3}$ is not adjacent to a vertex in any
component~$C_{h'}$, 
where $h' < i$, since in that case a shorter path would exist.
However, since vertices cannot be unselected as the online algorithm proceeds,
it follows that in~$G_{l_3}$, $v_{l_3}$ is not dominated by any vertex,
which is a contradiction.
Thus, $m\le 2$ and selecting just one additional vertex at the arrival of $v_{i_j}$
connects $C_i$ to an earlier component,
and the result follows inductively.
 
To see that the bound is tight, observe that the optimal incremental
connected dominating set 
of $P_n$ has $n-1$ vertices, while for even $n$, there is an
incremental dominating set of size $n/2$ with $n/2$ components.

To obtain~(\ref{thm:cds:i}),
consider an algorithm $\ALG$ processing vertices greedily,
while always selecting all 
vertices from $S$. That is, $v_1$ and all vertices of $S$ are always selected,
and when a vertex $v \not\in S$ arrives, it is selected
 if and only if it is not dominated by already selected vertices, in which case
it is called a \emph{bad} vertex.
Clearly, \ALG produces an \incr dominating set, $R''$, of $G$. 

To prove the upper bound on $|R''|$, we gradually mark components of $S$.
For a bad vertex $v_i$, let $v$ be
a vertex from $S$ dominating $v_i$, and let $C$ be the component of $S$
containing $v$.
Mark $C$.
To prove the claim it suffices to show that each component of $S$ can
be marked at most once, since each bad vertex leads to some component
of $S$ being marked.

Assume for the sake of contradiction that some component, $C$, of
$S$ is marked twice.
This happens because a vertex $v$ of $C$ is adjacent to a bad vertex $b$,
and a vertex $v'$ (not necessarily different from $v$) of $C$ is
adjacent to some later bad vertex $b'$. 
Since $G$ is always-connected and $b'$ was bad, 
$b$ and $b'$ are connected by a path not
including $v'$.
Furthermore, $v$ and $v'$ are connected by a path in $C$. Thus, the
edges $\SET{b,v}$ and 
$\SET{b',v'}$ imply the existence of a cycle in $G$,
contradicting the fact that it is a tree.

To see that the bound is tight, let $v_1,\ldots, v_m$, $m \equiv 2 \pmod 6$,
be a path in the standard order.  Let $G$ be obtained from
$P_m$ by attaching $m$ pendant vertices (new vertices of degree~$1$) to
each of the vertices $v_2, v_5, v_8,\ldots, v_m$, where the pendant
vertices arrive in arbitrary order, though respecting that
$G$ should be always-connected.
Each minimum incremental dominating set of
$G$ contains each of the vertices $v_2, v_5, v_8,\ldots, v_m$,
the vertex $v_1$, and one of the vertices $v_{3i}$ and $v_{3i+1}$ for
each $i$, and thus it has size $2(m+1)/3$.  On the other hand, the
vertices $v_2, v_5, v_8,\ldots, v_m$ form a dominating set $S$ of $G$
with $\cc(S) = (m+1)/3$.
\qed\end{proof}

Theorem~\ref{thm:cds} is best possible
in the sense that none of the assumptions can be omitted.
Indeed, Proposition~\ref{prop:ds-onopt} implies that it is not even
possible to bound
the size of an \incr (connected) dominating set in terms of the size
of a (connected) dominating set, much less
to bound the size of an \incr connected dominating set in terms of the
size of a dominating set. 
Therefore, (\ref{thm:cds:c})~and (\ref{thm:cds:ic})~in Theorem~\ref{thm:cds}
cannot be combined even on bipartite planar graphs.
The situation is different for trees:
Proposition~\ref{prop:parent}~(\ref{parent:one}) essentially leverages the fact
that any connected dominating set $D$ on a tree can
be produced by an \incr algorithm without increasing the size of $D$.

\subsection{Trees}

For \ds and \cds, we let \aparent denote the following algorithm for trees.  The
algorithm selects the first vertex.  When a new vertex~$v$ arrives,
if $v$ is not already dominated by a previously arrived vertex,
then the parent vertex that $v$ is adjacent to is added to the
dominating set.
Note that \aparent accepts all internal nodes of the tree rooted
at the first vertex, creating an incremental connected dominating set.
For \cds on trees, \aparent is $1$-competitive,
even against \offopt:

\begin{lemma}
\label{lemma:cds-trees-alg}
For \cds on any tree~$T$, 
$$
\aparent(T) \leq \left\{ 
\begin{array}{ll} 
\gamma_C(T) + 1, & \text{if } v_1 \text{ has degree $1$ in } T \\
\gamma_C(T), & \text{otherwise}.
\end{array}
\right.
$$
\end{lemma}
\begin{proof}
For \cds, \aparent selects no vertices of degree~$1$, except possibly $v_1$.
Thus, the algorithm selects all vertices of degree at least~$2$ plus at most one
vertex of degree~$0$ or~$1$.

For trees with at most two vertices, the minimum size of a connected dominating
set is~$1$.
For trees with more than two vertices, the minimum size of a
connected dominating set of any tree $T$ equals the number of vertices
with degree at least~$2$.
\qed\end{proof}

For \tds, \aparent is the same as for \ds and \cds, except that it 
 selects $v_1$ only if $v_2$ arrives, in which case it selects both $v_1$ and $v_2$.
Thus, \aparent for \tds selects at most one more vertex than \aparent
for \ds and \cds.
To show that for \tds on trees,
\aparent is $1$-competitive against \onopt,
we prove the following:

\begin{lemma}
\label{lemma:tds-is-connected}
For any \incr total dominating set $D$ for an always\--con\-nect\-ed graph~$G$,
all $D_i$ are connected.
\end{lemma}
\begin{proof}
For the sake of a contradiction, suppose that for some~$i$, 
$D_i$ is not connected, and let
$i$ be the smallest index with this property. It follows that the vertex
$v_i$ constitutes a singleton component of 
$D_i$. Thus, $v_i$ cannot be dominated by any other vertex of~$D_i$,
contradicting that the solution is \incr.
\qed\end{proof}

\begin{lemma}
\label{lemma:tds-trees-alg}
For \tds on any tree~$T$, $\aparent(T) = \onopt_T(T)$.
\end{lemma}
\begin{proof}
If $T$ consists of only one vertex, $\aparent(T) = \onopt_T(T) = 0$.
Otherwise, \aparent selects $v_1$, $v_2$, and all later internal
vertices.
\onopt also selects $v_1$ and $v_2$, and by
Lemma~\ref{lemma:tds-is-connected}, it has to select all internal
vertices.
Thus, the two algorithms select exactly the same set of vertices.
\qed\end{proof}

\begin{lemma}
\label{lemma:whole-tree}
For any online algorithm \A for \ds or \cds, there exist arbitrarily
large trees $T$, such that $\A(T) \geq n-1$.
\end{lemma}
\begin{proof}
We prove that the adversary can construct an arbitrarily large tree,
maintaining the invariant that at most 
one vertex is not included in the solution of \A.  The algorithm has
to select the first vertex, so the invariant holds initially.
When presenting a new vertex~$v_{i}$, the adversary
checks whether all vertices given so far are included in \A's
solution.  If this is the case, $v_{i}$ is connected to an arbitrary
vertex, and the invariant still holds.  Otherwise, $v_{i}$ is
connected to the unique vertex not included in $D_{i-1}$.
Now $v_{i}$ is not dominated, so \A must select an additional vertex.
\qed\end{proof}

\begin{proposition}
\label{prop:dsonlower}
For any online algorithm \A for \ds on trees, $\cron(\A) \ge 2$.
\end{proposition}
\begin{proof}
We argue that, for any always-connected bipartite graph, $G$,
we have that $\onopt(G) \leq \frac{n+1}{2}$.
Since trees are bipartite, the result then follows from Lemma~\ref{lemma:whole-tree}.
The smaller partite set~$S$ of any connected bipartite
graph $G$ is a dominating set of~$G$.
If the first presented vertex~$v_1$ belongs to~$S$, then $S$ is
an \incr dominating set of~$G$. Otherwise, $S\cup\SET{v_1}$ is an
\incr dominating set of~$G$.
\qed\end{proof}

The adversary strategy used in the proof of
Lemma~\ref{lemma:whole-tree} cannot give a
lower bound larger than~$2$ against \offopt, since the resulting tree
may not have any dominating set with fewer than $n/2$ vertices. 
Consider, for example, a caterpillar graph where each vertex of the
central path has exactly one neighbor not belonging to the central
path.

The following proposition concludes on the results for \ds, \cds, and
\tds on trees.

\begin{proposition}
\label{prop:parent}
For trees, the following hold.
\begin{enumerate}[(i)]
\item \label{parent:one}
For \ds, $\cron(\aparent) = 2$ \, and \, $\croff(\aparent) = 3$.
\item \label{parent:two}
For \cds, $\cron(\aparent) = \croff(\aparent) = 1$. 
\item \label{parent:three}
For \tds, $\cron(\aparent)=1$ \, and \, $\croff(\aparent)=2$.
\end{enumerate}
\end{proposition}
\begin{proof}
We prove~(\ref{parent:one}) first.
The lower bound on $\cron(\aparent)$ follows directly from
Proposition~\ref{prop:dsonlower}.
For the corresponding upper bound, note that
Lemma~\ref{lemma:cds-trees-alg} in combination with
Theorem~\ref{thm:cds}(\ref{thm:cds:ic}) 
imply that 
$\aparent(T) \leq \gamma_C(T) +1 \leq 2 \cdot \onopt(T)$,
for any tree $T$. 
The result on $\croff(\aparent)$ follows from
Theorem~\ref{thm:cds}(\ref{thm:cds:c}) and the proof that
Theorem~\ref{thm:cds}(\ref{thm:cds:c}) is tight. 

Item~(\ref{parent:two}) follows directly from Lemma~\ref{lemma:cds-trees-alg}.

In item~(\ref{parent:three}),
the result on $\cron(\aparent)$ follows directly from
Lemma~\ref{lemma:tds-trees-alg}.

For the upper bound on $\croff(\aparent)$, let $S$ be an optimal total
dominating set for a tree $T$. Assume that $|S| \geq 3$ and consider
the following calculations which we argue for below.
\begin{align*}
\aparent(T) 
& =    \onopt_T(T), \text{ by Lemma~\ref{lemma:tds-trees-alg}}\\
& \leq \onopt_C(T)\\ 
& \leq \offopt_C(T)+1\\ 
& \leq |S| + 2(c(S)-1) +1, \text{by
  Theorem~\ref{thm:cds}(\ref{thm:cds:c})}\\
& \leq 2|S|-1 
\end{align*}
The first inequality in the calculations above follows from the fact
that any connected dominating set of size at least $2$ is a total
dominating set, and since we assumed that an optimal total dominating set
for $T$ has at least three vertices, any connected dominating set for
$T$ must have at least two vertices.

The second inequality follows from Lemma~\ref{lemma:cds-trees-alg},
since \aparent is an incremental algorithm.

The last inequality follows from the fact that any connected component
in a total dominating set has at least two vertices.

For the lower bound on $\croff(\aparent)$, consider a path $v_1,
v_2,\ldots,v_{4n}$ for a positive integer $n$.
When given in the standard order, \aparent will select the first $4n-1$
vertices, whereas an optimal total dominating set is the set $\{v_{4i+2},
v_{4i+3} \mid 0 \leq i \leq n-1\}$ of size $2n$.
\qed\end{proof}

\subsection{Bipartite, bounded-degree, and general graphs}

We extend the \aparent algorithm for graphs that are not trees as
follows. When a vertex $v_i$, $i>1$, arrives, which is not
already dominated by one of the previously presented vertices,
\aparent selects \emph{any} of the neighbors of $v_i$ in $G_i$.
Again, it is easily seen that \aparent creates an incremental
connected dominating set.
We start with a few positive results for \aparent.

\begin{proposition}
\label{prop:cds-bip}
The following hold.
\begin{enumerate}[(i)]
\item \label{off-acbip-dscds}
 For \ds and \cds on always-connected bipartite graphs, for $n\geq 4$,
 $$\croff(\aparent)\leq n/2.$$
\item \label{off-ac-dscds}
 For \ds and \cds on always-connected graphs, for $n\geq 4$,
$$\croff(\aparent) \le n-2.$$
\item \label{off-tds}
 For \tds, $\croff(\aparent)\leq n/2$.
\end{enumerate}
\end{proposition}
\begin{proof}
For item~(\ref{off-acbip-dscds}), if $\gamma(G) \ge 2$, there is
nothing to prove. 
Therefore, we assume that there is a single vertex~$v$
adjacent to every other vertex. Since $G$ is bipartite, there is no 
edge between any of the vertices adjacent to~$v$, so $G$ is a star.
Since $G_i$ is connected for each $i$, the vertex $v$ arrives either
as the first or the second vertex. Furthermore, if another vertex arrives 
after $v$, then $v$ is selected by \aparent. Once $v$ is selected,
all future vertices are already dominated by $v$, so no
more vertices are selected, implying that $\aparent(G)\le 2$,
which concludes the proof. 

For item~(\ref{off-ac-dscds}), we  only need to consider the case of
$\gamma(G) = 1$,  since otherwise there is 
nothing to prove, and thus there is a vertex $v$ adjacent
to every other vertex of $G$.
Since after the arrival of any vertex, \aparent increases
the size of the dominating set by 
at most one, it suffices to prove that, immediately after some vertex has
been processed, there are two vertices not selected by \aparent.
First note that once $v$ is selected, \aparent does not select
any other vertex and thus we can assume that $v$ is not the first vertex.
Suppose that $v$ arrives after $v_i$, $i\geq 2$. The vertex $v_i$ has not yet
been selected when $v$ arrives, and $v$ is dominated by $v_1$, so there
are two vertices not selected.
The last remaining case is 
when $v$ arrives as the second vertex. In this case we distinguish whether
$v_3$ is adjacent to $v_1$, or not. If $v_3$ is adjacent to $v_1$,
then $v$ is not selected,
there are two vertices not selected ($v$ and $v_3$), and we are done.
If $v_3$ is not
adjacent to $v_1$, then \aparent selects $v$ when $v_3$ arrives.
No further vertex will be added to the dominating set, concluding the proof.

For any graph with at least one edge, any total dominating set contains
at least two vertices.
Thus, if \aparent selects more vertices than $\offopt_T$, $\offopt_T$ selects
at least two vertices.
This proves~(\ref{off-tds}).
\qed\end{proof}

The following result shows that
Proposition~\ref{prop:cds-bip}(\ref{off-ac-dscds}) is tight.

\begin{proposition}
For any online algorithm, \A, for \ds or \cds on always-connected
planar graphs, $\croff(\A) \geq n-2$.
\end{proposition}
\begin{proof}
By Lemma~\ref{lemma:whole-tree}, the adversary can construct a tree on
$n-1$ vertices, such that any online algorithm selects at least $n-2$
vertices.
If the adversary then gives one vertex connected to all $n-1$ vertices
in the tree, this last vertex constitutes a connected dominating set.
It is not difficult to see that any such graph is indeed planar.
\qed\end{proof}

For \ds,
let \greedy be the algorithm that selects an
arriving vertex, if and only it is not dominated by a previously
selected vertex.

\begin{proposition}
\label{prop:delta-upper}
For graphs of maximum degree $\Delta$, the following hold.
\begin{enumerate}[(i)]
\item \label{delta:dscds} For any algorithm \ALG for \ds or \cds, $\croff(\A) \le \Delta+1$.
\item \label{delta:greedy} For \ds, $\croff(\greedy) \le \Delta$.
\item \label{delta:tds} For any algorithm \ALG for \tds, $\croff(\A) \le \Delta$.
\item \label{delta:cdscon} For any algorithm $\ALG$ for \cds on connected graphs, $\croff(\A) \le \Delta -1$.
\end{enumerate}
\end{proposition}
\begin{proof}
For \ds and \cds, each vertex can only dominate itself and its at most $\Delta$
neighbors. Thus, $\gamma_C(G) \geq \gamma(G)\ge n/ (\Delta +1)$,
proving item~(\ref{delta:dscds}).

For item~(\ref{delta:greedy}), consider a dominating set $S = \{s_1, s_2,
\ldots, s_k\}$ of size $k=\gamma(G)$.
Partition the vertices of $G$ into $k$ sets $V_1, V_2, \ldots, V_k$
such that $s_i \in V_i$ and all vertices in $V_i \setminus \{s_i\}$
are dominated by $s_i$.
Clearly, $|V_i| \leq \Delta+1$ and 
if $V_i$ has $\Delta+1$ vertices, it is called a {\em critical} set.
If there are exactly $d$ critical sets, then
$n \leq d(\Delta+1) + (k-d)\Delta$.
Thus, $\gamma(G)=k \geq (n-d)/\Delta$.

For each critical set $V_i$, each vertex in the set is connected to at least
one other vertex.
Thus, if \greedy selects the $\Delta$ first vertices of $V_i$, it will
not select the last vertex of $V_i$.
This shows that, from each critical set, \greedy will select at most
$\Delta$ vertices.
Hence, $\greedy(G) \leq n -d$, concluding the proof of
item~(\ref{delta:greedy}). 

For \tds, a vertex can only dominate its at most $\Delta$
neighbors. Thus, $\gamma_T(G)\ge n/ \Delta$, proving
item~(\ref{delta:tds}).

For item~(\ref{delta:cdscon}), let $D$ be a minimum connected
dominating set of a connected graph $G$ with $|D| = k$.
The sum of the degrees of vertices in $D$ is bounded by
$k\Delta$ which is then also an upper bound on how many vertices
$D$ can dominate outside $D$.
Since $D$ is connected, 
any spanning tree of $D$ contains $k-1$ edges
and each endpoint is adjacent to the other endpoint in the spanning tree.
Thus, no vertices outside $D$ are dominated via these
edges.
Thus, at most 
$k\Delta - (2k-2)$ vertices not in $D$ can be dominated by $D$,
giving $n \leq k\Delta-k+2 = k(\Delta-1)+2$ vertices in $G$.
It follows that $\gamma_C(G) \ge (n-2)/(\Delta -1)$ and thus, for any
algorithm \ALG for \cds, $\croff(\A) \le \Delta -1$.
\qed\end{proof}

The upper bound of
Proposition~\ref{prop:delta-upper}(\ref{delta:greedy}) is almost
tight, even for always-connected bounded-degree 
graphs:

\begin{proposition}
For any online algorithm \ALG for \ds on always-connected
bounded-degree graphs, $\croff(\ALG) \ge \Delta-2$.
\end{proposition}

\begin{proof}
We adapt the construction in the proof of
Lemma~\ref{lemma:whole-tree} to work for bounded-degree graphs. 
The adversary first gives $n_1$ vertices inducing a tree.
For convenience, we let $n_1$ be a multiple of $\Delta$.

For the first $n_1$ vertices, the adversary uses the following
strategy.
If there is a vertex $v \not\in D_{i-1}$ with degree less than
$\Delta-1$, $v_{i}$ is connected to $v$.
Otherwise, $v_{i}$ is connected to any vertex $v \in D_{i-1}$ with degree less than
$\Delta-1$.
Thus, the following invariant is maintained.
At most one vertex $v \not\in D_i$ has degree less than $\Delta-1$.
After the first $n_1$ vertices, $n_2=n_1/\Delta$ vertices are given such
that each of the first $n_1$ vertices is adjacent to exactly one of
the last $n_2$ vertices.

Let $V_1$ be the set consisting of the first $n_1$ vertices, and let
$V_2$ contain the last $n_2$ vertices.
By construction, each vertex in $V_1
\setminus D$, except at most one, has at least $\Delta-1$ neighbors in
$V_1$, and for any pair of neighbors, $u,v \in V_1$, at least
one of $u$ and $v$ is included in $D$.
Thus, there are more than $(|V_1 \setminus D|-1)(\Delta-1)$ edges
between $V_1 \setminus D$ and $V_1 \cap D$.
Together with the fact that the number of edges in the subgraph induced
by $V_1$ is $n_1-1$,
this means that $(|V_1\setminus D| -1)(\Delta - 1) \le n_1 - 1$,
implying $|V_1 \setminus D| \leq (n_1-1)/(\Delta-1)+1$.
Thus, $|D| \ge |D\cap V_1| = n_1 - |V_1\setminus D| \ge n_1 - (n_1-1)/(\Delta -1)-1
> (\Delta-2)n_1/(\Delta-1)-1$.
Since $V_2$ constitutes a dominating set of size $n_1/\Delta$, this proves that
the asymptotic competitive ratio satisfies 
$\croff(\ALG) \ge \Delta(\Delta-2)/(\Delta-1) > \Delta-2$.
\qed\end{proof}

Our next aim is to show that there exists an algorithm which is
$n/4$-competitive against $\onopt$ on every graph.
Later, in Propositions~\ref{prop:bip-2layers} and~\ref{prop:bip-2layers-tds}, we prove that this is optimal.
For the algorithm, we use \emph{layers} in a graph~$G$.
The function $L$ assigns layer numbers to vertices as follows:
If $v_i$ has no neighbors when it arrives, let $L(v_i) = 1$; otherwise, let
$$L(v_i) = 1 + \min 
  \SETOF{ L(v_j)}{ \text{$v_j$ is a neighbor of $v_i$ in $G_i$}}.$$

The
algorithm, denoted \alayersparent, is a specialization of \aparent.
For each vertex~$v_i$, $i>1$,
if $v_i$ is not dominated by one of the already selected
vertices, it selects a neighbor of $v_i$ with the smallest
layer number. 
For \cds, if the vertex $v_i$
connects two or more connected components, the algorithm also adds a
minimum-sized set of 
vertices to $D$ to make it connected. This will include the current
vertex and at most one neighbor in each component being connected.
Furthermore, for \ds and \cds, the algorithm also adds the first
vertex to arrive in each of layers~$3$ and~$5$.

The pseudocode for \alayersparent for \ds and \cds is given in
Algorithm~\ref{alg:LPdscds}.

\begin{algorithm} 
\DontPrintSemicolon
\SetNoFillComment
\SetAlgoNoLine
\SetArgSty{textrm}
{\normalsize
\BlankLine
   $D \gets \emptyset$\;
   \While{a vertex $v_i$ is presented}{
      \eIf{$v_i$ has no neighbors in $G_i$}{
         $L(v_i) \gets 1$\;
         $D \gets D \cup \{ v_i\}$\;
         \label{first-layer-choose}
      }{
         $L(v_i) \gets 1+\min\{L(v_j) \mid v_j~\mbox{\rm is a neighbor
         of}~v_i~\mbox{\rm in}~G_i\}$\\
            \If{there is no  $v_j\in D$ such that $v_j$ dominates $v_i$}{
 	       Choose a neighbor $v_j$ of $v_i$ with $L(v_j) = L(v_i)-1$ \;
	       $D \gets D \cup \{ v_j\}$\;
               \label{main-choose-line}
            }
         \If{the problem is \cds}{ \label{cds:start}
            \If{$v_i$ connects vertices belonging to different
              connected components in $G_{i-1}$}{
               Add a minimum-sized set of vertices to $D$
               connecting the corresponding components of $D$ 
               \label{cds:end}
            }
         }
      }
         \If{$L(v_i)\in \{ 3,5\}$
            \label{if-secondary-choose-line}}{
            \If{$|\{ v_j\in G_i \mid L(v_j)=L(v_i)\} |=1$}{
               $D \gets D \cup \{ v_i\}$
               \label{secondary-choose-line}
             }
         }
      }
   }
 \caption{Algorithm \alayersparent for \ds and \cds.\label{alg:LPdscds}}
\end{algorithm}

The pseudocode for \alayersparent for \tds is given in Algorithm~\ref{alg:LPtds}.
Algorithm~\ref{alg:LPtds} is obtained from Algorithm~\ref{alg:LPdscds} by omitting lines~\ref{first-layer-choose} and
\ref{cds:start}--\ref{secondary-choose-line}
and adding the following (lines~\ref{tds:start}--\ref{tds:end}):
For each vertex in layer~$1$, its first neighbor $v$ to arrive is added
to $D$.

\begin{algorithm} 
\DontPrintSemicolon
\SetNoFillComment
\SetAlgoNoLine
\SetArgSty{textrm}
{\normalsize
\BlankLine
   $D \gets \emptyset$\;
   \While{a vertex $v_i$ is presented}{
      \eIf{$v_i$ has no neighbors in $G_i$}{
         $L(v_i) \gets 1$\;
      }{
            $L(v_i) \gets 1+\min\{L(v_j) \mid v_j~\mbox{\rm is a neighbor
            of}~v_i~\mbox{\rm in}~G_i\}$\\
            \If{there is no  $v_j\in D$ such that $v_j$ dominates $v_i$}{
 	       Choose a neighbor $v_j$ of $v_i$ with $L(v_j) = L(v_i)-1$ \;
	       $D \gets D \cup \{ v_j\}$\;
            }
            \If{$v_i$ has an undominated neighbor \label{tds:start}}{
               $D \gets D \cup \{ v_i \}$ \label{tds:end}
            }
         }
      }
   }
 \caption{Algorithm \alayersparent for \tds.\label{alg:LPtds}}
\end{algorithm}

We prove that \alayersparent is asymptotically optimal in most cases.
We consider \ds and \cds first.

\begin{lemma}
\label{lemma:layers}
Consider a graph $G$ and an incremental algorithm \ALG for \ds or \cds.
For each connected component, $H$, of the subgraph $G_i$ of $G$, the following hold.
\begin{enumerate}[(i)]
\item \label{layers:one}
 \ALG selects all vertices of the first layer of $H$.
\item \label{layers:two}
  For any two consecutive layers, $j$ and $j+1$ of $H$, if no
  vertices in layer $j$ are included in the final solution, the first
  vertex of layer $j+1$ is selected by \ALG.
\item \label{layers:three}
  If $H$ has at least $2k+1$ layers, $k \in \ints$, \ALG accepts at least
  $k+1$ vertices in $H$.
\end{enumerate}
\end{lemma}
\begin{proof}
Item~(\ref{layers:one}) follows immediately from the fact that each vertex in
layer 1 is isolated when it arrives. 

For item~(\ref{layers:two}), note that when the first vertex $v$ of layer $j+1$
arrives, it is only connected to vertices in layer $j$, and hence it
is not dominated. 
Since \ALG does not select any vertices from layer $j$, $v$ must be
selected. 

Item~(\ref{layers:three}) follows directly from items~(\ref{layers:one})
and~(\ref{layers:two}).
\qed\end{proof}

\begin{theorem}
\label{thm:upperdscds}
For \ds and \cds,
$\cron(\aparentplus) \leq (n+3)/4$.
\end{theorem}
\begin{proof}
First, if $\onopt(G) \geq 4$, then $\alayersparent(G) \leq n
\leq \frac{n}{4} \onopt(G)$.
Furthermore, if $\onopt(G) = 1$, then $\alayersparent(G) = \onopt(G)$.
Thus, we need only consider graphs, $G$, with $2 \leq \onopt(G) \leq 3$.

We distinguish several cases according to the 
number, $\ell$, of layers of $G$.
If $\ell \leq 2$, then $\alayersparent(G) = \onopt(G)$.
If $\ell \geq 7$, then by Lemma~\ref{lemma:layers}, $\onopt(G) \geq 4$.
Hence, we only need to consider the range $3 \leq \ell \leq 6$.

We consider \ds first.
For $i \geq 1$, let $n_i$ denote the size of the $i$th layer
and $s_i$ the number of vertices in the $i$th layer
selected by $\alayersparent$ in Line~\ref{first-layer-choose}
or~\ref{main-choose-line}
(thus, \emph{not} including the selections in Line~\ref{secondary-choose-line}).
Note that $s_{\ell}=n_{\ell+1}=0$.

Since each vertex in layer $i+1$ causes
at most one vertex in layer $i$ to be selected,
\[
s_i \le n_{i+1}, \text{ for } i \geq 2, 
\mbox{~~and~~}
s_i \le n_i, \text{ for } i \geq 1.
\]
From these two inequalities independently, we get
\[
\sum_{i=2}^{\ell-1} \frac{i-1}{\ell-1}s_i \leq 
\sum_{i=2}^{\ell-1} \frac{i-1}{\ell-1}n_{i+1} =
 \sum_{i=3}^{\ell} \frac{i-2}{\ell-1}n_{i}
\;\; \text{ and } \;\;
\sum_{i=2}^{\ell-1} \frac{\ell-i}{\ell-1}s_i \leq \sum_{i=2}^{\ell-1} \frac{\ell-i}{\ell-1}n_{i}\,.
\]
Adding these two inequalities, we obtain
$$\sum_{i=2}^{\ell-1} s_i \leq \sum_{i=2}^{\ell}\frac{\ell-2}{\ell-1}n_i\,.$$

Let $n'$ be the total number of vertices selected in
lines~\ref{first-layer-choose} and~\ref{secondary-choose-line}.
If $3 \leq \ell \leq 4$, then $n'=n_1+1$.
Finally, if $\ell \geq 5$, then $n'=n_1+2$.

Since $\alayersparent(G) = n' + \sum_{i=2}^{\ell-1} s_i$, we get
\begin{equation}
\label{eq:sum}
\alayersparent(G)-n'
=  \sum_{i=2}^{\ell-1} s_i
\leq \frac{\ell-2}{\ell-1}\sum_{i=2}^{\ell}n_i
= \frac{\ell-2}{\ell-1}(n - n_1)\,.
\end{equation}

We consider always-connected graphs first, for which $n_1=1$. 

For $\ell=3$, Inequality~(\ref{eq:sum}) yields $\alayersparent(G)\leq
(n-1)/2+2 = (n+3)/2$.
Since $\onopt(G) \geq 2$, $\alayersparent(G) / \onopt(G) \leq (n+3)/4 $. 

For $\ell=4$,
 Inequality~(\ref{eq:sum}) gives
$\alayersparent(G)\leq 2(n-1)/3+2 = (2n+4)/3$. 
If $\onopt(G) = 3$, then
$\alayersparent(G) / \onopt(G) \leq (2n+4)/9 < 
(n+3)/4$.
If $\onopt(G)=2$, it follows from Lemma~\ref{lemma:layers} that the
vertices selected by \onopt are the first vertices
in layers~$1$ and~$3$.
Since these vertices are selected on arrival 
by \alayersparent as well, 
\alayersparent selects the same vertices as \onopt, plus
a parent of the first vertex in layer~$3$.
Thus, it selects $3/2 \onopt$ vertices.
This ratio is smaller than $(n+3)/4$, since $n\geq 4$.

For $\ell=5$, Inequality~(\ref{eq:sum}) yields 
$\alayersparent(G)\leq 3(n-1)/4+3 = (3n+9)/4$. 
By Lemma~\ref{lemma:layers}, $\onopt(G) \geq 3$ and hence,
$\alayersparent(G) / \onopt(G) \leq (3n+9)/12 = (n+3)/4$. 

For $\ell=6$, it follows from Lemma~\ref{lemma:layers} that $\onopt(G)
= 3$ and the vertices selected by \onopt are the first vertices in
layers~$1$, $3$, and~$5$.
Since these vertices are selected on arrival
by \alayersparent as well, 
\alayersparent selects the same vertices as \onopt, plus
a parent of the first vertex in layer~$3$ and a parent of the first
vertex in layer~$5$.
Thus, it selects $5/3 \onopt$ vertices.
This ratio is smaller than $(n+3)/4$, since $n\geq 6$.

We now consider graphs which are not always-connected.
Note that we still assume that $\onopt(G) \le 3$, and
Inequality~(\ref{eq:sum}) still holds.
If $G$ is given in a disconnected order, layer~$1$ contains at least
two vertices and by Lemma~\ref{lemma:layers}, \onopt, just as
\alayersparent, accepts all vertices of layer~$1$. 
Therefore, if $\onopt(G)=2$, then $\ell \leq 2$, and
$\alayersparent(G)=\onopt(G)$. 
Moreover, if layer~$1$ contains three vertices,
then $\ell \leq 2$, and $\onopt(G) = 3 = \alayersparent(G)$. 
Hence, we only need to consider the case where $\onopt(G)=3$ and the
first layer contains exactly two vertices.
Note that, in this case, $3 \leq \ell \leq 4$.

If $\ell=3$, then by Inequality~(\ref{eq:sum}),
$\alayersparent(G) \leq (n-2)/2+3 = (n+4)/2$.
It follows that $\alayersparent(G)/\onopt(G) \leq (n+4)/6 < (n+3)/4$.

If $\ell=4$, then by Inequality~(\ref{eq:sum}),
$\alayersparent(G) \leq 2(n-2)/3+3 = (2n+5)/3$.
It follows that $\alayersparent(G)/\onopt(G) \leq (2n+5)/9 < (n+3)/4$.

We now consider \cds.
If the graph is always-connected, \alayersparent for \cds selects the same
vertices as for \ds, so the calculations for \ds also hold for \cds.
Thus, we only need to consider graphs that are not always-connected.

If layer~$1$ contains three vertices, then the graph cannot have an
 incremental \cds with fewer than four vertices, contradicting 
$\onopt(G) \leq 3$.
Thus, we can assume that the graph never has more than two connected
components and that the two components arrive in an always-connected
manner. 
If the two components remain unconnected, the above analysis for \ds
holds for each component.
Otherwise, 
\onopt connects the two components by
selecting exactly one vertex, $v$, and since \onopt is \incr, the
two components must be unconnected until the arrival of $v$.
Thus, the vertices selected by \alayersparent are exactly the three
vertices selected by \onopt and the first vertex of layer~$3$, if it
arrives.
Hence, $\alayersparent(G)/\onopt(G) \leq 4/3 < \frac{n+3}{4}$, since
$n \geq 3$.
\qed\end{proof}

We now consider \tds.
For general graphs, we obtain an upper bound of approximately $n/4$,
as we did for \ds and \cds.
For always-connected graphs, the upper bound is improved to
approximately $2n/9$.

\begin{theorem}
For \tds, $\cron(\aparentplus) \leq (n+2)/4$, and for \tds on
always-connected graphs, $(2n+1)/9 \leq \cron(\aparentplus) \leq
(2n+2)/9$.
\end{theorem}
\begin{proof}
We use the same notation as in the proof of
Theorem~\ref{thm:upperdscds}.
Thus, $n_i$ denotes the size of the $i$th layer, $s_i$ denotes the
number of vertices in the $i$th layer selected by $\alayersparent$,
and $\ell$ is the total number of layers.

Since the first vertex in each layer $i>1$ is only connected to vertices in
layer $i-1$, and choosing that first vertex does not dominate it, any 
incremental algorithm must choose at least one vertex
in each layer, except the last.
Hence, $\onopt(G) \ge \ell-1$.
Since $\alayersparent(G) \leq n$, this means that for general graphs,
we can assume $\ell \leq 4$.
For always-connected graphs, we can assume $\ell \leq 5$, since $\ell
\geq 6$ implies a ratio of at most $n/5 < 2n/9$.

If $\ell=1$, $\alayersparent(G)=\onopt(G)=0$.
Hence, it suffices to consider $\ell \geq 2$.

We use inequalities similar to those for \ds:
\begin{alignat*}{2}
& s_1 \leq n_1 \; & \text{ and } \; & s_1 \leq n_2\\
& s_2 \leq n_2 \; & \text{ and } \; & s_2 \leq n_1 + n_3\\
& s_i \leq n_i \; & \text{ and } \; & s_i \leq n_{i+1}, \text{ for } i
  \geq 3
\end{alignat*}

Before using the inequalities, we strengthen the inequality $s_2 \leq
n_1+n_3$.
When the first vertex of layer~$2$ arrives, \onopt as well as
\alayersparent will select this vertex and a vertex from layer~$1$.
If \alayersparent selects one more vertex from layer~$2$, \onopt will
also have to select an additional vertex, and hence, $\onopt(G) \geq 3$.
Thus,
 $$s_2=1, \text{ if } \onopt(G) = 2 \,.$$
Note that no vertex in layer~$1$ has a neighbor outside of layer~$2$.
Consider a vertex $u$ in layer~$1$.
When the first neighbor, $v$, of $u$ arrives, any incremental
algorithm has to select $v$.
Thus, the vertices in layer~$2$ that \alayersparent selects in
order to dominate vertices in layer~$1$ are also selected by \onopt.
Hence, since \onopt selects at least one vertex in layer~$1$, 
 $$s_2 \leq \onopt(G)-1 + n_3\,.$$

We consider general graphs first.
Recall that for general graphs, we only need to consider $2 \leq \ell
\leq 4$.

For $\ell = 2$, 
\begin{align*}
\alayersparent(G)
&  = s_1+s_2 \\
&  \leq \left( \frac12 n_1 + \frac12 n_2 \right) + 
       \left( \onopt(G)-1+n_3 \right) \\
&  = \frac12(n-2) + \onopt(G), \text{ since } n_3=0.
\end{align*}
Hence, 
$\alayersparent(G) / \onopt(G) \leq \frac14(n-2) + 1 =
\frac14(n+2)$.

For $\ell=3$, we first consider the case $\onopt(G)=2$.
In this case, $s_2=1$.
Hence, 
$$\alayersparent(G) 
  = s_1 + s_2 
  \leq \left( \frac12 n_1 + \frac12 n_2 \right) + 1 
  < \frac12n +1\,.
$$
and 
$\alayersparent(G) / \onopt(G) < \frac14(n+2)$.
For $\onopt(G)=3$, we note that 
\begin{align*}
\alayersparent(G) 
&  = s_1+s_2\\
&  \leq \left( \frac{3}{4}n_1 + \frac{1}{4}n_2 \right) +
       \left( \frac{1}{2}n_2 + \frac{1}{2}(\onopt(G)-1+n_3) \right) \\
&  < \frac{3}{4}n + \frac12{\onopt(G)}\,.
\end{align*}
Thus, 
$\alayersparent(G) / \onopt(G) < \frac14 n + \frac12 = \frac14(n+2)$.

For $\ell=4$, $\onopt \geq 3$
and
\begin{align*}
&  \alayersparent(G) \\
&  = s_1+s_2+s_3 \\
&  \leq \left( \frac34 n_1 + \frac14 n_2 \right) + 
       \left( \frac12 n_2 + \frac12(\onopt(G)-1+n_3) \right) + 
       \left( \frac14 n_3 + \frac34 n_4 \right) \\
&  < \frac34 n + \frac12 \onopt(G)\,.
\end{align*}
Thus, $\alayersparent(G) / \onopt(G) < \frac14 n + \frac12 = \frac14(n+2)$.

We now consider always-connected graphs, for which $s_1=n_1=1$.
We have argued that for the upper bound, it is sufficient to consider
$2 \leq \ell \leq 5$. 

If $\ell=2$, $\alayersparent(G)=\onopt(G)=2$.

For $\ell=3$, we consider $\onopt(G)=2$ first.
In this case, $\alayersparent(G) = s_1 + s_2 = 1+1 = \onopt(G)$.
For $\onopt(G) \geq 3$, note that 
\begin{align*}
\alayersparent(G) 
&  = s_1+s_2\\
&  \leq 1 + \left( \frac{1}{2}n_2 + \frac12(n_3+1) \right) \\
&  = 1 + \frac12 n\\
&  = \frac{n+2}{2}\,.
\end{align*}
Thus, $\alayersparent(G)/\onopt(G) \leq (n+2)/6 < (2n+2)/9$, since $n > 2$.

If $\ell=4$, then $\onopt(G) \geq 3$.
Moreover,
\begin{align*}
\alayersparent(G) 
&  = s_1+s_2+s_3 \\
&  \leq 1 + \left( \frac23 n_2 + \frac13(n_3+1) \right) + 
       \left( \frac13 n_3 + \frac23 n_4 \right) \\
&  = \frac43 + \frac23 (n-1)\\
&  = \frac{2n+2}{3} \,.
\end{align*}
Thus, $\alayersparent(G) / \onopt(G) \leq \frac{2n+2}{9}$.

If $\ell=5$, then $\onopt(G) \geq 4$.
Moreover,
\begin{align*}
&  \alayersparent(G) \\
&  = s_1+s_2+s_3+s_4 \\
&  \leq 1 + \left( \frac34 n_2 + \frac14(n_3+1) \right) + 
       \left( \frac12 n_3 + \frac12 n_4 \right) +
       \left( \frac14 n_4 + \frac34 n_5 \right) \\
&  = \frac54 + \frac34 (n-1)\\
&  = \frac{3n+2}{4} \,.
\end{align*}
Thus, $\alayersparent(G) / \onopt(G) \leq \frac{3n+2}{16} <
\frac{2n}{9}$, since $n \geq 5$.

Finally, we prove the lower bound for always-connected graphs, using
the following adversarial input sequence defining a graph $G$ with four
layers.
The first layer consists of the vertex $u$.
The following three layers each have $m$ vertices, for some large
integer $m$.
The vertices of the second, third, and fourth layers are called
$v_1,\ldots,v_m$, $w_1,\ldots,w_m$, and $x_1,\ldots,x_m$, respectively.
The vertices are given layer by layer, in the order according to their
numbering.

No vertex in the second layer is connected to any other vertex in the
same layer.
In the third layer, $w_1$ is connected to $v_1$, and for $2 \leq i
\leq m$, $w_i$ is connected to $v_i$ and $w_1$.
In the fourth layer, for $1 \leq i \leq m-1$, $x_i$ is connected to
$w_1$ and $w_{i+1}$, and $x_m$ is connected $w_1$.

\onopt selects the three vertices $u$, $v_1$, and $w_1$.

\alayersparent selects $u$ and $v_1$ on arrival.
For $2 \leq i \leq m$, it selects $v_i$ when $w_i$ arrives.
Hence, when the first three layers have arrived, all vertices of
layers~1 and 2 have been selected.
Each vertex $x_1, \ldots, x_{m-1}$ can be dominated by either $w_1$ or
$w_{i+1}$.
If \alayersparent always chooses the latter, it will select all
vertices $w_2,\ldots,w_{m}$, and when $x_m$ arrives, it must select
$w_1$.
In total, \alayersparent selects $1+2m = 1 + 2(n-1)/3 = (2n+1)/3$,
yielding a ratio of $\alayersparent(G)/\onopt(G) = (2n+1)/9$. 
\qed\end{proof}

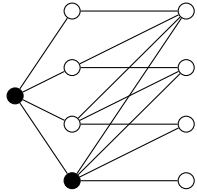
\begin{figure}[!htb]
\begin{center}
\begin{tikzpicture}[scale=0.75]
	\vertex[fill] (l0) at (0,1.5) {};
	\vertex[fill] (m0) at (1,0) {};
	\vertex       (m1) at (1,1) {};
	\vertex       (m2) at (1,2) {};
	\vertex       (m3) at (1,3) {};
	\vertex       (r0) at (3,0) {};
	\vertex       (r1) at (3,1) {};
	\vertex       (r2) at (3,2) {};
	\vertex       (r3) at (3,3) {};

	\path
		(l0) edge (m0)
		(l0) edge (m1)
		(l0) edge (m2)
		(l0) edge (m3)
		(m0) edge (r0)
		(m0) edge (r1)
		(m0) edge (r2)
		(m0) edge (r3)
		(m1) edge (r1)
		(m1) edge (r2)
		(m1) edge (r3)
		(m2) edge (r2)
		(m2) edge (r3)
		(m3) edge (r3)
	;
\end{tikzpicture}
\caption{A three-layer construction;
the minimum connected dominating set is indicated by the black vertices
(Proposition~\ref{prop:bip-2layers}).}
\label{fig:two_layers}
\end{center}
\end{figure}

\begin{proposition}
\label{prop:bip-2layers}
On bipartite graphs, the following hold for any online algorithm \A
for \ds or \cds.
\begin{enumerate}[(i)]
\item \label{bip:n}
 For \ds and \cds on always-connected graphs, $\cron(\A)\ge n/4$.
\item \label{bip:deltaac}
 For \ds on always-connected bounded-degree graphs, $$\cron(\A) \ge
 \Delta/2-1/4.$$
\item \label{bip:deltaaccds}
 For \cds on always-connected bounded-degree graphs, $$\cron(\A) \ge \Delta/3.$$ 
\item \label{bip:delta}
 For \cds on bounded-degree graphs, $\cron(\A)\ge \Delta/2$.
\end{enumerate}
\end{proposition}
\begin{proof}
For items~(\ref{bip:n}) and~(\ref{bip:delta}), we prove that
for any integer $\Delta\ge 2$, 
there is an always-connected bipartite graph, $G$, with maximum degree~$\Delta$
such that $\A(G) \geq \Delta = n/2$ and
$\onopt(G)=\onoptC(G) = 2$.

The graph $G$ consists of three layers.
The first layer contains only one vertex $u$, and
the second layer contains $\Delta-1$ vertices $v_1,
\ldots, v_{\Delta-1}$ adjacent to $u$.
After the entire second layer is presented to
the algorithm, the vertices of the second layer are indistinguishable
to the algorithm.
The last layer consists of $\Delta-1$ vertices $w_1,\ldots, w_{\Delta-1}$,
which will be given in that order,
with adjacencies as follows: For
$i=1,\ldots, \Delta-1$, $w_i$ is
connected to $\Delta - i$ vertices of the second layer in such a way
that 
$N(w_{i+1}) \subset N(w_{i})$
and $N(w_{i})$ contains as few vertices
from $D_{i-1}$ as
possible.
An example of this construction for $\Delta = 4$ is depicted in
Figure~\ref{fig:two_layers}.

Consider the situation when the vertex $w_i$ arrives.
If the set~$N(w_i)$ does not contain a vertex
from $D_{i-1}$, then \A must select
at least one additional vertex at this time.
Thus, \A{} selects at least $\Delta-1 = (n-1)/2$ vertices
from the second and third layer, plus the root.
Since there is a vertex~$v$ in the second layer that is adjacent to all
vertices in the third layer, $\{u, v\}$ is an \incr connected and total 
dominating set of $G$, concluding the proof of~(\ref{bip:n}).
Since the adversary can use any number of copies of $G$, this also
finishes the proof of~(\ref{bip:delta}).

For items~(\ref{bip:deltaac}) and~(\ref{bip:deltaaccds}),
note that the adversary can
use any number of copies of $G$, with one vertex in the third layer of
copy $k$ connected to the vertex in the first layer of copy $k+1$.

For \ds, note that for any given algorithm $\A$, $G$ is constructed in such 
a way that $\A$ must select at least
one vertex from layers~$1$ and~$2$ and at least $\Delta-1$ vertices from
layers~$2$ and $3$ in each copy of $G$.
Thus, if \A selects all $\Delta-1$ vertices of layer~$3$ in some copy
of $G$, it selects at least $\Delta$ vertices from this copy.
Otherwise, the adversary can connect the next copy of $G$ to a vertex $w$
not selected by \A.
In this case, the algorithm will have to select $w$ or the first
vertex of the next copy of $G$. 
Hence, from two consecutive copies of $G$, \A selects at least
$2\Delta-1$ vertices.
On the other hand, choosing two vertices from each copy as described above will
result in an incremental dominating set. 
This proves~(\ref{bip:deltaac}).

For \cds, the adversary will connect adjacent copies of $G$ in the
following way.
The vertex in layer~$3$ connected to all vertices in layer~$2$ will be
connected to the first vertex of the following copy of $G$.
Thus, an incremental connected dominating set can be created by
selecting one vertex from each of layers~$1$ and~$2$ as described
above plus the vertex in layer~$3$ connected to the next copy of $G$.
Again, \A will select at least $\Delta-1$ vertices from layers $2$
and $3$ in each copy of $G$, and to make the dominating set connected,
it will also select the vertex in layer~$1$.
This proves~(\ref{bip:deltaaccds}). 
\qed\end{proof}

The above adversary strategy does not work for \tds, since \onopt
needs to accept the two first vertices of the graph $G$.
Thus, we use a slightly different graph to prove the following
proposition.

\begin{proposition}
\label{prop:bip-2layers-tds}
On bipartite graphs, the following hold for any online algorithm \A.
\begin{enumerate}[(i)]
\item \label{bip:tdsac}
 For \tds, on always-connected graphs, $$\cron(\A)\ge n/6 \text{~~and~~}
 \cron(\A) \ge \Delta/3.$$
\item \label{bip:tds}
 For \tds, $\cron(\A)\ge n/4$ \; and \; $\cron(\A)\ge (\Delta+1)/2$.
\end{enumerate}
\end{proposition}
\begin{proof}
For item~(\ref{bip:tdsac}), we use a graph, $G'$, identical to the
graph $G$ used in the proof of Proposition~\ref{prop:bip-2layers},
except that the second layer has $\Delta$ vertices, and no vertex in
layer~$3$ is connected to the first vertex of layer~$2$.
For bounded-degree graphs, the adversary gives many copies of $G'$,
and for each copy except the 
last, the first vertex of layer~$2$ is connected to the first vertex
of layer~$2$ in the following copy.
In all copies of $G'$, except the first, the first vertex of layer~$2$
is given before the vertex of layer~$1$.

For each copy of $G'$, any \incr algorithm for \tds will select the vertex of
layer~$1$ and the first vertex of layer~$2$, and
\A will also select the remaining $\Delta-1$ vertices of layer~$2$.
Among the last $\Delta-1$ vertices selected by \A, \onopt will only
select the last one to be selected by \A.
This proves item~(\ref{bip:tdsac}).

For item~(\ref{bip:tds}), we use a graph consisting of only two layers.
The vertices of layer~$1$ are given first.
Then, the following is repeated.
As long as there is a vertex in layer~$1$ not selected by \A, a
vertex is given which is adjacent to exactly the vertices in layer~$1$
not yet selected by \A.
For each of these vertices, \ALG has to select a vertex in layer~$1$.
It follows that layer~$1$ contains $\Delta \geq n/2$ vertices, and \ALG
selects at least $\Delta+1$ vertices, all of those in layer~$1$ and the
first in layer~$2$.
On the other hand, \onopt chooses only the first vertex of layer~$2$
and the last vertex of layer~$1$ to be included in \ALG's dominating
set.
This proves~(\ref{bip:tds}). 
\qed\end{proof}

\section{The Cost of Being Incremental}
\label{sec:offline}
This section is devoted to comparing the performance of
\incr algorithms and \offopt.
Since \offopt{} performs at least as well as \onopt and \onopt
performs at least as well as any online algorithm, each lower bound
in Table~\ref{table:any-offopt} is at least the maximum of
the corresponding lower bound in Table~\ref{table:any-onopt} and
the corresponding lower bound for $\croff(\onopt)$.
Similarly, each upper bound in Table~\ref{table:any-onopt} is at most
the corresponding upper bound in
Table~\ref{table:any-offopt}.
In both cases, we mention only bounds
that cannot be obtained in this way from cases considered already.
We first give two positive results.

\begin{proposition}
For \ds, the following hold.
\begin{enumerate}[(i)]
\item \label{trees:two}
 On trees, $\croff(\onopt) \le 2$.
\item \label{ac:ntwo}
 On always-connected graphs, $\croff(\onopt) \le \lceil n/2 \rceil$.
\end{enumerate}
\end{proposition}
\begin{proof}
Item~(\ref{trees:two}) follows directly from
Theorem~\ref{thm:cds}(\ref{thm:cds:i}).

We now consider item~(\ref{ac:ntwo}).
For a fixed ordering of the vertices of $G$, consider the layers
$L(v)$ assigned
to vertices of $G$. It is easy to see that the set of vertices
in the odd layers is an \incr solution for \ds and similarly
for the set of vertices in even layers plus the vertex $v_1$. Therefore,
\onopt can select the smaller of these two sets, which necessarily
has at most $\lfloor (n-1)/2 \rfloor +1 = \lceil n/2 \rceil$ vertices.
\qed\end{proof}

The remaining results are negative results.

\begin{proposition}
\label{prop:ds-onopt}
On bipartite planar graphs, the following hold.
\begin{enumerate}[(i)]
\item \label{dso:one} For \ds,
$\croff(\onopt) \ge \Delta$ \, and \, 
$\croff(\onopt) \ge n-1$.
\item \label{dso:two} For \cds,
$\croff(\onopt) \ge \Delta+1$ \, and \, 
$\croff(\onopt) \ge n$.
\end{enumerate}
\end{proposition}
\begin{proof}
We prove that for each $\Delta \ge 3$, $i > 0$, and $n=i(\Delta+1)$, there is a
bipartite planar graph~$G$ with $n$ vertices and maximum
degree~$\Delta$ such that 
\begin{align*}
&\onopt(G) = \frac{\Delta}{\Delta+1}n\,,\;
 \onoptC(G) = n\,, \text{ and}\\
&\gamma(G) = \gamma_C(G) = \frac{n}{\Delta+1}\,,
\end{align*}
implying the first lower bound of both~(\ref{dso:one}) and~(\ref{dso:two}).
Letting $i=1$, and hence $n=\Delta+1$, gives the second lower bound of
both~(\ref{dso:one}) and~(\ref{dso:two}).

Let $G$ consist of $i$ disjoint copies of the star on $\Delta +1$ vertices,
with the center of each star arriving as the last vertex among
the vertices of that particular star.
Clearly, $\gamma(G) = \gamma_C(G) = n/(\Delta+1)$.
On the other hand, any \incr dominating set has to contain
every vertex, except the last vertex of each star, since all
these vertices are pairwise non-adjacent. 
In addition, any \incr connected dominating set
has to contain the centers of the stars
to preserve connectedness of the solution in each component.
It follows that for Dominating Set,
$\onopt$ selects $n\Delta/(\Delta + 1)$ vertices, and for Connected
Dominating Set, it selects all $n$ vertices. 
\qed\end{proof}

\begin{proposition}
\label{prop:ids-onopt}
For \ids on bipartite planar graphs, the following hold.
\begin{enumerate}[(i)]
\item \label{ids:acn}
 On always-connected graphs, $\croff(\onopt) \geq n-1$. 
\item \label{ids:delta}
 On bounded-degree graphs, $\croff(\onopt) \geq \Delta$.
\item \label{ids:acdelta}
 On always-connected bounded-degree graphs, $\croff(\onopt) \ge \Delta-1$.
\end{enumerate}
\end{proposition}
\begin{proof}
For~(\ref{ids:acn}), let $G$ be a star, where the second vertex to
arrive is the center vertex.
Clearly, $\gamma_I(G) = 1$. 
Since the first vertex is always selected by any \incr
algorithm, the center vertex 
cannot be selected. Consequently, all $n-1$ vertices of 
degree $1$ have to be selected in the dominating set, which proves
the lower bound of the first part.

For~(\ref{ids:delta}), note that the adversary can
give any number of copies of $G$.

For~(\ref{ids:acdelta}), note that the adversary can
make arbitrarily many copies of $G$ and connect two consecutive copies
by identifying two vertices of degree $1$, one from each copy.
\qed\end{proof}

\begin{proposition}
\label{prop:ids-onopt-upper}
For \ids, \, $\croff(\onopt) \leq \Delta \leq  n-1$
\end{proposition}
\begin{proof}
To prove the upper bound of~$\Delta$, consider any graph, $G$, with
maximum degree~$\Delta$, and
let $S = \{s_1,\ldots, s_k\}$
 be an independent dominating 
set of $G$ of size $k=\gammaI(G)$.

Let $R_1, \ldots, R_k$ be a partition of $V$ such that all vertices in $R_i$ 
are dominated by $s_i$. Let $R'_i = R_i \setminus \{s_i\}$ and note that 
$R'_1,\ldots, R'_k$ is a partition of $V \setminus \{s_1,\ldots, s_k\}$.
For each $i$, the vertex $s_i$ can be in an independent dominating set~$D$
only if $R'_i \cap D = \emptyset$.
Thus, $|D| \leq \sum_{i=1}^k \max\{ |\{ s_i\}|, |R'_i|\}
= \sum_{i=1}^k \max\{ 1, |R'_i|\}$, and
$|D|/|S|$ is bounded by the maximum possible size of $R'_i$, which
is $\Delta$.
Since $\Delta\leq n-1$ for all simple graphs, this concludes the proof.
\qed\end{proof}

\begin{lemma}
\label{lemma:path}
For any positive integer $n \ge 3$ and $P_n$ given in the standard order, 
$$\onopt(P_n) = \lceil n/2 \rceil \: \text{ and } \: \onopt_C(P_n) = \onopt_T(P_n) =
n-1\,.$$
\end{lemma}
\begin{proof}
The result for $\onopt$ follows from
Lemma~\ref{lemma:layers}(\ref{layers:three}) and
the fact that selecting the vertices with odd index results in an
incremental dominating set.

For $\onopt_C$, note that $v_1$ must be selected and hence, each
$v_i$, $2 \leq i \leq n-1$, must be selected no later than when
$v_{i+1}$ arrives.

The result on $\onopt_T$ follows from
Lemma~\ref{lemma:tds-is-connected} and the result on $\onopt_C$. 
\qed\end{proof}

A \emph{fan} of degree $\Delta$ is the graph obtained from a path $P_{\Delta}$ 
by addition of a vertex $v$ that is adjacent to all
vertices of the path, as in Figure~\ref{fig:fan}. The adversarial 
order of a fan is defined by the standard order of the underlying path,
followed by the vertex~$v$.

\begin{figure}[!htb]
\begin{center}
\begin{tikzpicture}[scale=0.75]
	\vertex       (p0) at (0,0) {};
	\vertex       (p1) at (1,0) {};
	\vertex       (p2) at (2,0) {};
	\vertex       (p3) at (3,0) {};
	\vertex       (t0) at (1.5,1) {};

	\path
		(p0) edge (p1)
		(p1) edge (p2)
		(p2) edge (p3)
		(t0) edge (p0)
		(t0) edge (p1)
		(t0) edge (p2)
		(t0) edge (p3)
	;
\end{tikzpicture}
\caption{A fan of degree~$4$ (Proposition~\ref{prop:fans}).
\label{fig:fan}}
\end{center}
\end{figure}
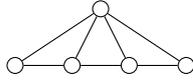

\begin{proposition}
\label{prop:fans}
For always-connected planar graphs, the following hold.
\begin{enumerate}[(i)]
\item \label{fans:ds} For \ds, $\croff(\onopt) \ge n/2$.
\item \label{fans:cds} For \cds, $\croff(\onopt) \ge n-2$.
\item \label{fans:tds} For \tds, $\croff(\onopt) \ge n/2 -1$.
\end{enumerate}
\end{proposition}
\begin{proof}
Let $G$ be a fan of degree $\Delta=n-1$, where $n$ is even, given in the adversarial
order. By Lemma~\ref{lemma:path},
$\onopt(G) = n/2$ and $\onoptC(G) = \onoptT(G) = n-2$.
Furthermore,
$\gamma(G) = \gamma_C(G) = 1$, and $\gamma_T(G) = 2$,
since $v_n$ forms a connected dominating set of size $1$, and 
$\{v_1,v_n\}$ is a total dominating set of size $2$.
This proves~(\ref{fans:ds})--(\ref{fans:tds}).
\qed\end{proof}

An \emph{alternating fan} with $k$ fans of degree $\Delta$
consists of $k$ copies of the fan of degree $\Delta$,
where the individual copies are joined in a path-like manner by
identifying some of the vertices of degree~$2$, as in
Figure~\ref{fig:alt-fan}. 
Thus, $n=k(\Delta+1)-(k-1)$ and $k=(n-1)/\Delta$.
The adversarial 
order of an alternating fan is defined by the concatenation of
the adversarial orders of the underlying fans.

\begin{figure}[!htb]
\begin{center}
\begin{tikzpicture}[scale=0.75]
	\vertex       (p0) at (0,1) {};
	\vertex       (p1) at (1,1) {};
	\vertex       (p2) at (2,1) {};
	\vertex       (p3) at (3,1) {};
	\vertex       (p4) at (4,1) {};
	\vertex       (p5) at (5,1) {};
	\vertex       (p6) at (6,1) {};
	\vertex       (p7) at (7,1) {};
	\vertex       (p8) at (8,1) {};
	\vertex       (p9) at (9,1) {};
	\vertex[fill] (t0) at (1.5,2) {};
	\vertex[fill] (b0) at (4.5,0) {};
	\vertex[fill] (t1) at (7.5,2) {};

	\path
		(p0) edge (p1)
		(p1) edge (p2)
		(p2) edge (p3)
		(p3) edge (p4)
		(p4) edge (p5)
		(p5) edge (p6)
		(p6) edge (p7)
		(p7) edge (p8)
		(p8) edge (p9)
		(t0) edge (p0)
		(t0) edge (p1)
		(t0) edge (p2)
		(t0) edge (p3)
		(b0) edge (p3)
		(b0) edge (p4)
		(b0) edge (p5)
		(b0) edge (p6)
		(t1) edge (p6)
		(t1) edge (p7)
		(t1) edge (p8)
		(t1) edge (p9)
	;
\end{tikzpicture}
\caption{An alternating fan with $3$ fans of degree~$4$ (Proposition~\ref{prop:bridges}(\ref{bridges:ds})).}
\label{fig:alt-fan}
\end{center}
\end{figure}
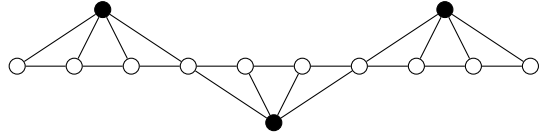

A \emph{bridge of degree $\Delta$ with $k$ sections} is
obtained from a path of $k(\Delta-2)$ vertices $v_1, v_2, \ldots,
v_{k(\Delta-2)}$, in that order, together with $k$ vertices $u_1, u_2, \ldots,
u_k$.
For $1 \leq i \leq k$, $u_i$ is connected to the $\Delta-2$ vertices
$$v_{(i-1)(\Delta-2)+1}, v_{(i-1)(\Delta-2)+2},\ldots, v_{i(\Delta-2)},$$ and
for $1 \leq i \leq k-1$, $u_i$ is connected to $u_{i+1}$.
See Figure~\ref{fig:conn-fan} for an example.
The adversarial order of a bridge of degree $\Delta$ with $k$ sections
is $v_1$, $v_2$, \ldots, $v_{k(\Delta-2)}$, $u_1$, $u_2$, \ldots, $u_k$. 

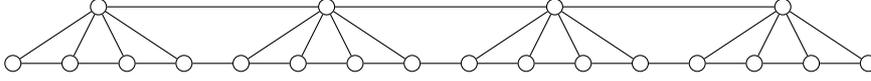
\begin{figure}[!htb]
\begin{center}
\begin{tikzpicture}[scale=0.75]
	\vertex       (p0) at (0,0) {};
	\vertex       (p1) at (1,0) {};
	\vertex       (p2) at (2,0) {};
	\vertex       (p3) at (3,0) {};
	\vertex       (p4) at (4,0) {};
	\vertex       (p5) at (5,0) {};
	\vertex       (p6) at (6,0) {};
	\vertex       (p7) at (7,0) {};
	\vertex       (p8) at (8,0) {};
	\vertex       (p9) at (9,0) {};
	\vertex       (p10) at (10,0) {};
	\vertex       (p11) at (11,0) {};
	\vertex       (p12) at (12,0) {};
	\vertex       (p13) at (13,0) {};
	\vertex       (p14) at (14,0) {};
	\vertex       (p15) at (15,0) {};
	\vertex       (t0) at (1.5,1) {};
	\vertex       (t1) at (5.5,1) {};
	\vertex       (t2) at (9.5,1) {};
	\vertex       (t3) at (13.5,1) {};

	\path
		(p0) edge (p1)
		(p1) edge (p2)
		(p2) edge (p3)
		(p3) edge (p4)
		(p4) edge (p5)
		(p5) edge (p6)
		(p6) edge (p7)
		(p7) edge (p8)
		(p8) edge (p9)
		(p9) edge (p10)
		(p10) edge (p11)
		(p11) edge (p12)
		(p12) edge (p13)
		(p13) edge (p14)
		(p14) edge (p15)
		(t0) edge (t1)
		(t1) edge (t2)
		(t2) edge (t3)
		(t0) edge (p0)
		(t0) edge (p1)
		(t0) edge (p2)
		(t0) edge (p3)
		(t1) edge (p4)
		(t1) edge (p5)
		(t1) edge (p6)
		(t1) edge (p7)
		(t2) edge (p8)
		(t2) edge (p9)
		(t2) edge (p10)
		(t2) edge (p11)
		(t3) edge (p12)
		(t3) edge (p13)
		(t3) edge (p14)
		(t3) edge (p15)
	;
\end{tikzpicture}
\caption{A bridge of degree~$6$ with $4$ sections (Proposition~\ref{prop:bridges}(\ref{bridges:cds})).}
\label{fig:conn-fan}
\end{center}
\end{figure}

For even $k$, a \emph{modular bridge of degree $\Delta$ with $k$ sections}
is the same as a bridge of degree $\Delta-1$ with $k$ sections, except
that for even $i$, the edge between $u_i$ and $u_{i+1}$ is not present.

\begin{figure}[!htb]
\begin{center}
\begin{tikzpicture}[scale=0.75]
	\vertex       (p0) at (0,0) {};
	\vertex       (p1) at (1,0) {};
	\vertex       (p2) at (2,0) {};
	\vertex       (p3) at (3,0) {};
	\vertex       (p4) at (4,0) {};
	\vertex       (p5) at (5,0) {};
	\vertex       (p6) at (6,0) {};
	\vertex       (p7) at (7,0) {};
	\vertex       (p8) at (8,0) {};
	\vertex       (p9) at (9,0) {};
	\vertex       (p10) at (10,0) {};
	\vertex       (p11) at (11,0) {};
	\vertex       (p12) at (12,0) {};
	\vertex       (p13) at (13,0) {};
	\vertex       (p14) at (14,0) {};
	\vertex       (p15) at (15,0) {};
	\vertex       (t0) at (1.5,1) {};
	\vertex       (t1) at (5.5,1) {};
	\vertex       (t2) at (9.5,1) {};
	\vertex       (t3) at (13.5,1) {};

	\path
		(p0) edge (p1)
		(p1) edge (p2)
		(p2) edge (p3)
		(p3) edge (p4)
		(p4) edge (p5)
		(p5) edge (p6)
		(p6) edge (p7)
		(p7) edge (p8)
		(p8) edge (p9)
		(p9) edge (p10)
		(p10) edge (p11)
		(p11) edge (p12)
		(p12) edge (p13)
		(p13) edge (p14)
		(p14) edge (p15)
		(t0) edge (t1)
		(t2) edge (t3)
		(t0) edge (p0)
		(t0) edge (p1)
		(t0) edge (p2)
		(t0) edge (p3)
		(t1) edge (p4)
		(t1) edge (p5)
		(t1) edge (p6)
		(t1) edge (p7)
		(t2) edge (p8)
		(t2) edge (p9)
		(t2) edge (p10)
		(t2) edge (p11)
		(t3) edge (p12)
		(t3) edge (p13)
		(t3) edge (p14)
		(t3) edge (p15)
	;
\end{tikzpicture}
\caption{A modular bridge of degree~$5$ with $4$ sections
(Proposition~\ref{prop:bridges}(\ref{bridges:tds})).}
\label{fig:total-fan}
\end{center}
\end{figure}
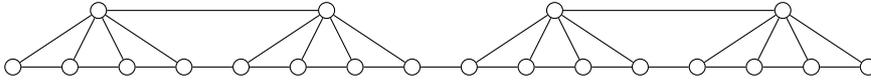

\begin{proposition}
\label{prop:bridges}
For always-connected bounded-degree planar graphs, the following hold.
\begin{enumerate}[(i)]
\item \label{bridges:ds} For \ds, $\croff(\onopt) \ge (\Delta-1)/2$.
\item \label{bridges:cds} For \cds, $\croff(\onopt) \ge \Delta-2$.
\item \label{bridges:tds} For \tds, $\croff(\onopt) \ge \Delta-1$.
\end{enumerate}
\end{proposition}
\begin{proof}
For~(\ref{bridges:ds}), let $G$ be an alternating fan with $k$ fans of degree $\Delta$, for any 
$\Delta \ge 4$, given in the adversarial order.
We prove that 
$\onopt(G) > n(\Delta-1)/(2\Delta)$ and $\gamma(G) \leq (n - 1)/\Delta$.
Starting with the latter, a fan consists of $\Delta+1$ vertices,
but the fans share one vertex, so a new one starts every $\Delta$
vertices, except for the final vertex which accounts for the $-1$.
For the former claim,
in Figure~\ref{fig:alt-fan}, the vertices belonging to a dominating
set of size $k = (n-1)/\Delta$ are filled in (black).
Since, by Lemma~\ref{lemma:path}, any \incr dominating set
on a path $P$ in the standard order has at least $\CEIL{|V(P)|/2}$ vertices,
\onopt must select at least $\CEIL{(n-k)/2}$ vertices of $G$.
Inserting $k=(n-1)/\Delta$ into $(n-k)/2$ gives
$(n(\Delta-1)+1)/(2\Delta)$, resulting in a ratio larger than $(\Delta-1)/2$.

For~(\ref{bridges:cds}),
let $G$ be a bridge of degree $\Delta$ with $k$ sections, given in the 
adversarial order, and let $m=k(\Delta -2)$.
By Lemma~\ref{lemma:path}, we have $\onopt(G) \ge \onopt(P_m) = k(\Delta-2)-1$.
The last $k$ vertices form a connected dominating set of $G$ and, thus,
$\gamma_C(G) \le k$.

For~(\ref{bridges:tds}), 
let $G$ be a modular bridge of degree $\Delta$ with $k$ sections
given in the adversarial order.
Let $m=k(\Delta -1)$.
By Lemma~\ref{lemma:path}, we have $\onopt(G) \ge \onopt(P_m) =  k(\Delta-1)-1$.
Clearly, $\gamma_T(G) \le k$, and the result follows.
\qed\end{proof}

For any $n\ge 2$, a \emph{two-sided fan of size~$n$}
is the graph obtained from a
path on $n-2$ vertices by attaching two additional vertices, one to the
even-numbered vertices of the path and the other to the odd-numbered
vertices of the path. The two additional vertices are connected by
an edge. An adversarial order of a two-sided fan is defined by 
the standard order of the path, followed by the two additional
vertices in any order.
See Figure~\ref{fig:rotor} for an illustration of a two-sided fan of
size~$10$.

\begin{figure}[!htb]
\begin{center}
\begin{tikzpicture}[scale=0.75]
	\vertex       (p0) at (0,2) {};
	\vertex       (p1) at (1,2) {};
	\vertex       (p2) at (2,2) {};
	\vertex       (p3) at (3,2) {};
	\vertex       (p4) at (4,2) {};
	\vertex       (p5) at (5,2) {};
	\vertex       (p6) at (6,2) {};
	\vertex       (p7) at (7,2) {};
	\vertex       (bo) at (3.5,1) {};
	\vertex       (to) at (3.5,3) {};
	\path
		(p0) edge (p1)
		(p1) edge (p2)
		(p2) edge (p3)
		(p3) edge (p4)
		(p4) edge (p5)
		(p5) edge (p6)
		(p6) edge (p7)
		(bo) edge (p0)
		(bo) edge (p2)
		(bo) edge (p4)
		(bo) edge (p6)
		(to) edge (p1)
		(to) edge (p3)
		(to) edge (p5)
		(to) edge (p7)
	;
        \draw    (bo) to[out=0,in=0,distance=5.5cm] (to);
\end{tikzpicture}
\caption{A two-sided fan of size~$10$ (Proposition~\ref{prop:two-sided-fan}).
\label{fig:rotor}}
\end{center}
\end{figure}
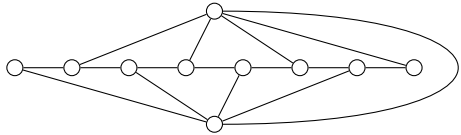

\begin{proposition}
\label{prop:two-sided-fan}
For both \cds and \tds
on always-connected bipartite planar graphs, we have
$\croff(\onopt) \ge (n-3)/2$.
\end{proposition}
\begin{proof}
Let $G_n$ be a two-sided fan of size~$n$, given in an adversarial order.
It suffices to prove that
$\onoptC(G_n) = \onoptT(G_n) = n-3$
and $\gamma(G_n) = \gamma_C(G_n) = \gamma_T(G_n) = 2$.
This is straightforward from the facts that the first $n-2$ vertices of~$G$
induce a path and any \incr connected or total dominating set on
$P_{n-2}$ given in the standard order has size at least~$n-3$.
\qed\end{proof}

\section{Conclusion and Open Problems}
Online algorithms for four variants of the dominating set problem
are analyzed using competitive analysis comparing to \onopt{} and \offopt{},
two reasonable alternatives for the optimal algorithm having
knowledge of the entire input.
Several graph classes are considered, and tight results are 
obtained in most cases.

The difference between \onopt and \offopt is that \onopt is
required to maintain an \incr solution (as any online algorithm),
while \offopt is only
required to produce a solution for the final graph.
The online algorithms are compared to both \onopt and \offopt,
and \onopt is compared to \offopt, in order to investigate
why all online algorithms tend to perform poorly against \offopt. 
Is this due only to the requirement to be \incr,
or is it more generally because of the lack of knowledge of the future?

Inspecting the results in the tables, perhaps the most striking conclusion
is that the competitive ratios of any online algorithm and \onopt,
respectively, against \offopt, are almost identical.
This indicates that the requirement to
maintain an \incr dominating set is a severe restriction, which can be 
offset by the full knowledge of the input only to a very small extent. 
On the other hand, when we restrict our attention to online algorithms
against \onopt,
it turns out that the handicap of not knowing the future still presents 
a barrier,
leading to competitive ratios of the order of~$n$ or $\Delta$
in most cases.

One could reconsider the nature of the irrevocable decisions, 
which originally stemmed from practical applications.
Which assumptions on irrevocability are relevant for practical 
applications, and which irrevocability components make the problem hard
from an online perspective?
We expect that these considerations will apply to many other online problems as well.

There is relatively little difference observed between three of the variants
of Dominating Set considered: Dominating Set, Connected Dominating
Set, and Total Dominating Set. In fact, the results for Total
Dominating Set generally followed directly from those for
Connected Dominating Set as a consequence of Lemma~\ref{lemma:tds-is-connected}.
The results for Independent Dominating Set were significantly different from
the others. It can be viewed as the minimum maximal
independent set problem since any maximal independent set is a
dominating set. This problem has been studied in the context of
investigating the performance of the greedy algorithm for the
independent set problem.
In fact, the unique incremental independent dominating set is the
set produced by the greedy algorithm for independent set.

In yet another orthogonal dimension, we compare the results for 
various graph classes.
Dominating Set is a special case of Set Cover and is notoriously 
difficult in classical complexity,
being NP-hard~\cite{K72}, $W[2]$-hard~\cite{DF95}, and not
approximable within $c\log n$ for any constant $c$
on general graphs~\cite{F98}.
On the positive side, on planar graphs, the problem is FPT~\cite{ABFKN02}
and admits a PTAS~\cite{B94}, 
and it is approximable within $\log \Delta$
on bounded-degree graphs~\cite{CH79}.
On the other hand, the relationship between
the performance of online algorithms
and structural properties of graphs
is not particularly well understood.
In particular, there are problems where
the absence of knowledge of the future is irrelevant; examples of such problems
in this work are \cds and \tds on trees, and \ids on any graph class.
As expected, for bounded-degree graphs, the competitive ratios 
are of the order of $\Delta$,
but closing the gap between $\Delta/2$ and $\Delta$
seems to require additional ideas.
On the other hand, for planar graphs, the problem, rather surprisingly, 
seems to be as difficult as the general case when compared to \offopt. 
When online algorithms for planar graphs
are compared to \onopt, we suspect there might be
an algorithm with constant competitive ratio.
At the same time, this case is the most notable open problem directly related 
to our results. Drawing inspiration from classical complexity,
one may want to eventually consider more specific graph classes in the 
quest for understanding exactly what structural properties
make the problem solvable.
From this perspective, our consideration of planar, bipartite, and
bounded-degree graphs is a natural first step. 

\section*{Acknowledgment}
The authors would like to thank an anonymous referee for
constructive suggestions.

\bibliographystyle{spmpsci}
\bibliography{refs}

\end{document}